\documentclass[a4paper]{article}

\usepackage{amsmath, amsfonts, amssymb, amsthm}
\usepackage{color}

\newtheorem{theorem}{Theorem}[section]
\newtheorem{lemma}[theorem]{Lemma}

\theoremstyle{definition}
\newtheorem{definition}[theorem]{Definition}
\newtheorem{example}[theorem]{Example}
\newtheorem{remark}[theorem]{Remark}

\usepackage{algorithm}
\usepackage{algorithmic}

\newcommand{\fq}{\mathbb{F}_q}

\date{}

\title{Decoding of Matrix-Product Codes\thanks{Keywords: Linear Codes, Matrix-Product Codes, Decoding Algorithm, Minimum Distance. Subclass: 94B05; 94B35.}}
\author{Fernando Hernando\thanks{Supported in part by Spanish MEC MTM2007-64704.}\\San Diego State University\\ \texttt{fhernando@mail.sdsu.edu} \and Diego Ruano\thanks{Supported in part by the Danish National Research Foundation and the National Science Foundation of China (Grant No.11061130539) for the Danish-Chinese Center for Applications of Algebraic Geometry in Coding Theory and Cryptography, and Spanish MEC MTM2007-64704.}\\Aalborg University\\ \texttt{diego@math.aau.dk}}

\begin{document}

\maketitle

\begin{abstract}
We propose a decoding algorithm for the $(u\mid u+v)$-construction that decodes up to half of the minimum distance of the linear code. We extend this algorithm for a class of matrix-product codes in two different ways. In some cases, one can decode beyond the error correction capability of the code.\end{abstract}

\section{Introduction}

Matrix-product codes, $[C_{1}\cdots C_{s}]\cdot A$, were introduced by Blackmore and Norton in \cite{Blackmore-Norton}. They may also be seen as a generalization of the $(u\mid u+v)$-construction.  Advantages of this method are, first, that long codes can be created from old ones and, second, that the parameters or the codes are known under some conditions \cite{Blackmore-Norton,hlr,Ozbudak}. Other generalizations include \cite{hr} and \cite{van}.

In \cite{hlr}, a decoding algorithm for matrix-product codes with $C_1 \supset \cdots \supset C_s$ was presented. In this work, we present an alternative to that algorithm, where we do not need to assume that the codes $C_1, \ldots , C_s$ are nested. In section 3, we present the new algorithm for $s=l=2$, $(u \mid u+v)$-construction, the main assumption that we should consider is $d_2 \ge 2 d_1$, where $d_i$ is the minimum distance of $C_i$, $d_i =d(C_i)$. The new algorithm decodes up to half of the minimum distance. Furthermore, if $d_1$ is odd and $d_2 > 2 d_1$, we are able to decode beyond this bound, obtaining just a codeword with a high probability.

From the algorithm in section \ref{se:2}  we derive two extensions for matrix-product codes defined using a matrix $A$, of arbitrary size $s \times l$, which verifies a certain property called non-singular by columns. The main difference between these two algorithms resides in the following fact: the algorithm is section \ref{se:gene1} requires stronger assumptions ($d_i \ge l d_1$, for all $i$) than the one in section \ref{se:gene2} ($d_i \ge i d_1$,  for all $i$), but it is computationally less intense. Both algorithms decode up to half of the designed minimum distance of the code \cite{Ozbudak}, that is known to be sharp in several cases \cite{Blackmore-Norton,hlr} (for intance if $C_1 \supset \cdots \supset C_s$). If $d_1$ odd and $l$ even, we can decode beyond this bound obtaining a list of codewords that will contain just one codeword with a high probability. The algorithm in section \ref{se:gene1} does not become computationally intense for large  $s,l$.

\section{Matrix-Product Codes}\label{sect:mpc}

A matrix-product code is a construction of a code from old ones.

\begin{definition}
Let $C_1, \ldots, C_s \subset \mathbb{F}_q^m$ be linear codes of length
$m$ and a matrix $A=(a_{i,j}) \in \mathcal{M}(\fq, s \times l)$, with
$s\leq l$. The matrix-product code $C=[C_1 \cdots C_s] \cdot A$
is the set of all matrix-products $[c_1 \cdots c_s] \cdot A$ where $c_i\in
C_i$ is an $m \times 1$ column vector $c_i=(c_{1,i},\ldots,c_{m,i})^T$ for $i=1,\ldots, s$. Therefore, a typical codeword $ p$ is

\begin{equation}\label{MatrixCodeword}
p= \left(
\begin{tabular}{ccc}
$c_{1,1}a_{1,1}+\cdots+c_{1,s}a_{s,1}$ & $\cdots$ & $c_{1,1}a_{1,l}+\cdots+c_{1,s}a_{s,l}$\\
$\vdots$ & $\ddots$ & $\vdots$\\
$c_{m,1}a_{1,1}+\cdots+c_{m,s}a_{s,1}$ & $\cdots$ & $c_{m,1}a_{1,l}+\cdots+c_{m,s}a_{s,l}$\\
\end{tabular}\right).
\end{equation}
\end{definition}

The $i$-th column of any codeword is an element of the form $\sum_{j=1}^s a_{j,i}
c_j\in \mathbb{F}_q^m$, therefore reading the entries of the $m\times l$-matrix above in column-major order, the codewords can be viewed as vectors of length $ml$,
\begin{equation}\label{VectorCodeword}
 p =\left(\sum_{j=1}^s a_{j,1} c_j, \ldots , \sum_{j=1}^s a_{j,l} c_j \right)
\in\mathbb{F}_q^{ml}.
\end{equation}

If $C_i$ is an $[m,k_i,d_i]$ code then one has that
$[C_1 \cdots C_s] \cdot A$ is a linear code over $\mathbb{F}_q$ with
length $lm$ and dimension $k=k_1+\cdots+k_s$  if the matrix $A$ has
full rank and $k < k_1+\cdots+k_s$ otherwise. 

Let us denote by $R_i= (a_{i,1},\ldots,a_{i,l})$ the element of
$\mathbb{F}_q^l$ consisting of the $i$-th row of $A$, for
$i=1,\ldots,s$. We denote by $D_i$ the minimum distance
of the code $C_{R_i}$ generated by $\langle R_1,\ldots,
R_i\rangle$ in $\fq^l$. In \cite{Ozbudak} the following lower bound for the minimum distance of the matrix-product code $C$ is obtained,
\begin{equation}\label{lowebound}
d(C)\geq d_C = \min\{d_1D_1,d_2D_2, \ldots ,d_s D_s\},
\end{equation}where $d_i$ is the minimum distance of $C_i$. If $C_1, \ldots, C_s$ are nested codes, $C_1 \supset \cdots \supset C_s$, the previous bound is sharp \cite{hlr}.

In \cite{Blackmore-Norton}, the following condition for the matrix $A$ is introduced.

\begin{definition}\cite{Blackmore-Norton}\label{de:nsc}
Let $A$ be a $s\times l$ matrix and  $A_t$ be the matrix consisting of the first $t$ rows of $A$. For $1\leq j_1< \cdots < j_t\leq l$, we denote by $A(j_1,\ldots,j_t)$ the $t\times t$ matrix consisting of the columns $j_1,\ldots,j_t$ of $A_t$.

A matrix $A$ is non-singular by columns if $A(j_1,\ldots,j_t)$ is non-singular for each $1\leq t\leq s$ and $1\leq j_1< \cdots < j_t\leq l$. In particular, a non-singular by columns matrix $A$  has full rank.
\end{definition}

Moreover, if $A$ is non-singular by columns, the bound $d_C$ in (\ref{lowebound}) is$$d(C) \ge d_C = \min\{ld_1,(l-1)d_2, \ldots ,(l-s+1)d_s \}$$and it is known to be sharp in several cases: it was shown in \cite{Blackmore-Norton} that if $A$ is non-singular by columns and triangular, (i.e. it is a column permutation of an upper triangular matrix), then the bound (\ref{lowebound}) for the minimum distance is sharp. Furthermore, if A is non-singular by columns and the codes $C_1 \ldots C_s$  are nested, then this bound (\ref{lowebound}) is also sharp.

A decoding algorithm for the matrix-product code $C=[C_1 \cdots C_s] \cdot A \subset \fq^{ml}$, with $A$ non-singular by columns and $C_1 \supset \cdots \supset C_s$ was presented \cite{hlr}, assuming that we have a decoding algorithm for $C_i$, for $i=1,\ldots,s$. We present in next section another decoding algorithm for a matrix-product code with $s=l=2$.

\section{A decoding algorithm for the $(u\mid u+v)$-construction}\label{se:2}

We consider a decoding algorithm for the $(u\mid u+v)$-construction, that is, a matrix-product code with $s=l=2$, $C = [C_1 C_2]\cdot A$ with $d_2 \ge 2 d_1$ and $d_1 \ge 3$, where $d_i = d(C_i)$ is the minimum distance of $C_i$. Let \begin{equation*} A=\left(
\begin{tabular}{cc}
1 & 1 \\
0 & 1 \\
\end{tabular}\right).
\end{equation*}Note that $C$ is the $(u|u+v)$-construction and that an equivalent code will be obtained with any matrix of rank $2$.

Let $t_1$ be the error-correction capability of $C_1$, $t_1 = \lfloor \frac{d_1 -1 }{2} \rfloor  \ge 1$, that is $d_1 = 2 t_1 + 1$ if $d_1$ is odd and $d_1 = 2 t_1 + 2$ if $d_1$ is even. The minimum distance of $C$ is $d(C)=\min \{2d_1 , d_2\}= 2 d_1$ \cite{masl}. Thus the error correction capability of the code $C$ is $$t=\left \lfloor\frac{2 d_1 - 1}{2}\right \rfloor = \left\{ \begin{array}{ll} 2 t_1 & ~ \mathrm{if} ~ d_1 ~ \mathrm{is~odd}\\
2 t_1 +1 & ~ \mathrm{if} ~ d_1 ~ \mathrm{is~even}
\end{array} \right.$$ 

We provide a decoding algorithm for the matrix-product code $C$, assuming that we have a decoding algorithm $DC_i$ for $C_i$ which decodes up to $t_i$ errors, for $i=1,2$. Let $r = p + e$ be a received word where $p \in C$ and the error vector $e$ verifies $wt (e) \le t$. Note that a typical word $p \in C$ is $[c_1 c_2]\cdot A=(c_1 , c_1 + c_2)$, namely a received word $r$ is $r = (r_1, r_2) = (c_1 + e_1 , c_1 + c_2 + e_2)$.

Consider $r_2 - r_1 = c_1 + c_2 + e_2 - c_1 -e_1 = c_2 + (e_2 - e_1)$. We may decode $r_2 -r_1$ using the decoding algorithm $DC_2$ to obtain $c_2$, since $c_2 \in C_2$ and $wt(e_2 -e_1) < d_2/2$ because $$wt(e_2 -e_1) \le wt (e_1) + wt (e_2) = wt (e) \le t < d_1  \le \frac{d_2}{2}.$$

Since we know $c_2$ we may consider $r^{2)}_2 = r_2 -c_2 = c_1 + e_2$ and let $r^{2)}_1 = r_1 = c_1 + e_1$. We claim that there exists $i_1 \in \{1,2\}$ such that $wt(e_{i_1}) \le t_1$: assume that such an $i$ does not exist, then $$wt (e) = wt (e_1 ) + wt (e_2) \ge 2 t_1 + 2,$$a contradiction. Let $wt(e_{i_1}) \le t_1$, then we can obtain $c_1$ by decoding $r^{2)}_{i_1}$ with the decoding algorithm $DC_1$. A priori, we do not know which index $i_1$ is, however we will be able to detect it by checking that we have not corrected more than $\lfloor (d(C)-1)/2\rfloor$ errors in total. That is, for $p=(c_1 , c_1 + c_2)$ and $p' = (c'_1 , c'_1 + c_2)$, we check whether $d(r,p) \le t$ and $d(r,p') \le t$. 	

\begin{remark}
Let us compare this decoding algorithm to the algorithm in \cite{hlr}. In the algorithm in \cite{hlr}, we assume that $C_1 \supset C_2$ and for this algorithm we assume that $ d_2 \ge 2 d_1$. Comparing the complexity of the algorithms: In the algorithm in \cite{hlr}, we should run $DC_1$ and $DC_2$ twice, in the worst case situation. For this algorithm, we run $DC_1$ twice and $DC_2$ once. Both algorithms decode up to the error-correction capability of the code.
\end{remark}

For $d_1$ odd and $d_2 > 2d_1$, the previous algorithm can also be used for correcting $t +1 = 2t_1 +1$ errors, that is, one more error than the error-correction capability of $C$. The algorithm outputs a list with one or two codewords, containing the sent word. Let us assume now that $wt(e) \le t+1$, again we may obtain $c_2$ by decoding $r_2 -r_1$ since $wt (e_2 -e_1) \le t_2$ because$$wt(e_2 -e_1) \le wt (e_1) + wt (e_2) = wt (e) \le t +1 = 2 t_1 +1 = d_1  < \frac{d_2}{2}.$$Again there will be an index $i_1 \in \{1,2\}$ such that $wt(e_{i_1}) \le t_1$ because otherwise $wt(e) \ge 2 t_1+2 > 2 t_1 + 1$. Hence, we also decode $r^{2)}_{i_2}$ using $DC_1$. Let $p=(c_1 , c_1 + c_2)$ and $p' = (c'_1 , c'_1 + c_2)$ as above, $d(p,r) \le t+1$ and $d(p',r) \ge t +1$. 

\begin{itemize}
\item $d( r,p)= wt( (c_1 + e_1 , c_1 + c_2 + e_2 ) -(c_1, c_1 + c_2) ) = wt (e) \le t  + 1$.

\item $d( r,p')= wt( (c_1 + e_1 , c_1 + c_2 + e_2 ) -(c'_1, c'_1 + c_2) ) = wt (c_1 - c'_1 + e_1 , c_1 - c'_1 + e_2) \ge   2 d_1 - wt(e) \ge  2 ( 2t_1+1) - (2t_1 + 1) = 2t_1 + 1= t+1 $.
\end{itemize}If we have that $d(p,r) , d(p',r) \le t+1$ we output both codewords, in other case we output only $p$. Note that the probability of having two codewords in the output list is negligible, since $d(r,p') = t+1$  if and only if $d(c_1 , c'_1)=d_1$ and for every $e_{j,i} \neq 0$, with $j=1,\ldots, m$, $i=1,2$, one has that $e_{j,i} = - (c_{j,i} -c'_{j,i})$.

We will consider in this article two different extensions of this algorithm for any $s$ and $l$, with $s \le l$. Namely, for the particular case where $s=l=2$, both extensions are equal.

\section{A decoding algorithm for Matrix-Product codes, first extension}\label{se:gene1}

In this section we propose an extension of the algorithm in the previous section for matrix-product codes with any $s \le l$, the algorithm in this section is less computationally intense than the algorithm in \cite{hlr} for large $s,l$. In the following section we will propose another extension. Let  $C= [C_1 \cdots C_s] \cdot A$ be a matrix-product code, with $d_i \ge l d_1$, for $i=2,\ldots,s$, and $d_1 \ge 3$, where $d_i = d(C_i)$ is the minimum distance of $C_i$.  We also require that  $A$ is non-singular by columns. 

The error-correction capability of $C_i$ is $t_i = \lfloor \frac{d_i -1 }{2} \rfloor \ge 1$. From (\ref{lowebound}), one has that the designed minimum distance of $C$ is $d(C) \ge d_C =\min \{ld_1, (l-1) d_2, \ldots , (l - s +1)  d_l\}= l d_1$. Hence, the designed error correction capability of the code $C$ is $$t=\left \lfloor\frac{l d_1 - 1}{2}\right \rfloor = \left\{ \begin{array}{ll} l t_1 + \lfloor \frac{l-1}{2} \rfloor & ~ \mathrm{if} ~ d_1 ~ \mathrm{is~odd}\\
l t_1 + l -1 & ~ \mathrm{if} ~ d_1 ~ \mathrm{is~even,}
\end{array} \right.$$ since $d_1 = 2 t_1 + 1$ if $d_1$ is odd and $d_1 = 2 t_1 + 2$ if $d_1$ is even.

We provide a decoding algorithm for the matrix-product code $C$ that decodes up to half of its designed minimum distance, assuming that we have a decoding algorithm $DC_i$ for $C_i$ which decodes up to $t_i$ errors, for $i=1,\ldots ,s$. A codeword in $C$ is an $m\times l$ matrix which has the form $ p = [c_1,\ldots,c_s] \cdot A = (\sum_{j=1}^s a_{j,1} c_j, \ldots , \sum_{j=1}^s a_{j,l} c_j )$, where $c_j \in C_j$, for all $j$. We denote by $p_i = \sum_{j=1}^s a_{j,i} c_j   \in \fq^{m}$ the $i$-th block of $p$, for $i=1,\ldots,l$. Suppose that $ p $ is sent and that we receive $r =  p  +  e $, where $ e =(e_1,e_2,\ldots,e_l)$  is an error vector, an $m\times l$ matrix, with weight $wt (e) \le t$.

Let $B$ be a matrix in $\mathcal{M}(\fq, l \times s)$, such that $A B$ is the $s\times s$-identity matrix. Such a matrix exists because $A$ has rank $s$ and it can be obtained by solving a linear system, but it is not unique if $s<l$. Let $w_i = (0 , \ldots , 0 ,1 , 0 , \ldots ,0)^T \in \fq^s$ be the vector that has all  coordinates equal to zero, excepting the $i$-th coordinate that  is equal to $1$. For  $i \in  \{2 , \ldots , s\}$, consider $v_i = (v_{1,i}, \ldots , v_{l,i})^T \in \fq^l$ equal to $v_i = B w_i$. One has that  $p  v_i = \sum_{j=1}^l v_{j,i} p_j = c_i$, since $p v_i = p B w_i = [c_1, \ldots , c_s] w_i= c_i$. Therefore$$r v_i=\sum_{j=1}^l v_{j,i} r_j = \sum_{j=1}^l v_{j,i} p_j + \sum_{j=1}^l v_{j,i} e_j= c_i + \sum_{j=1}^l v_{j,i} e_j.$$

For $i=2, \ldots, s$, we can decode $r v_i$ with the decoding algorithm $DC_i$ to obtain $c_i$, since $c_i \in C_i$ and $$wt \left(\sum_{j=1}^l v_{j,i} e_j \right)  \le \sum_{j=1}^l wt(e_j) = wt (e) \le t = \left\lfloor\frac{l d_1 - 1}{2}\right\rfloor \le \left\lfloor\frac{d_i - 1}{2}\right\rfloor =t_i$$

As we have already computed $c_2, \ldots, c_s$ we may consider now $r'_i = r_i - \sum_{j=2}^s a_{i,j} c_j = a_{1,i} c_1 + e_i$, for $i=1,\ldots,l$. We claim that there exists $i \in \{ 1, \ldots , l\}$ such that $w(e_i) \le t_1$ because if $wt(e_i ) > t_1$ for all $i$ then $wt(e) \ge lt_1 + l >t$. Therefore, we correct $r'_1/a_{1,1} , \ldots , r'_l/a_{1,l}$, with $DC_1$ and at least one of them gives $c_1$ as output. Note that $a_{1,i} \neq 0$, for $i=1, \ldots, l$ since $A$ is non-singular by columns. We have $l$ candidates for $c_1$, $c_1^{i)}=DC_1(r'_j/a_{1,i})$,  for $i=1,\ldots ,l$,  we can detect which candidate is equal to $c_1$ by checking that we have not corrected more than $\lfloor (d-1)/2\rfloor$ errors in total, that is, we check whether $d(r - [c^{i)}_1, c_2 \ldots , c_s] \cdot A) \le \lfloor (d-1)/2\rfloor $, for $i=1,\ldots,l$.

The algorithm is outlined as a whole in procedural form in Algorithm \ref{alg:dec1}.

\newcommand{\deccc}{\ensuremath{\mbox{\sc A decoding algorithm for $C=[C_1 \cdots C_s] \cdot A$, first extension}}}
\begin{algorithm}[h!]
  \caption{\deccc}
\algsetup{indent=2em}
  \begin{algorithmic}[1]\label{alg:dec1}
    \REQUIRE Received word $r = p + e$ with $c \in C$ and $wt(e) < d(C)/2$, where $d_i = d(C_i)$ with $l d_1 < d_i$ and $A$  full rank. Decoder $DC_i$ for code $C_i$, $i = 1,\ldots ,s$.
    \ENSURE $p$.
    \medskip
    \STATE $r'=r$;
    \STATE Find $B$, a right inverse of $A$ ($AB=Id$);    
    \FOR{$i=2, \ldots, s$}
         \STATE $v = B e_i$;
         \STATE $c_i = DC_i (rv)$;
    \ENDFOR
    \STATE $r=(r_1 - \sum_{j=2}^s a_{j,1} c_j, \ldots , r_l - \sum_{j=2}^s a_{j,l} c_j)$; 
    \FOR{$i=1, \ldots, l$}\label{li:indi}
         \STATE $c_1 = DC_1 (r_i/a_{1,i})$;
         \IF{$c_{1}=$ {\it ``failure"}}
                \STATE {\it Break} the loop and consider next $i$ in line \ref{li:indi};
         \ENDIF  
         \STATE $p=[c_1\cdots c_s] \cdot A$;
             \IF{$p \in C$ and $wt(r'-p ) \le \lfloor (d(C)-1)/2 \rfloor $ }\label{li:check}
             \RETURN  $p$;
             \ENDIF
    \ENDFOR

  \end{algorithmic}
\end{algorithm}

\begin{remark}
Let us compare this decoding algorithm to the algorithm in \cite{hlr}. In both algorithms we assume that $A$ is non-singular by columns. For the algorithm in this section, we assume that $l d_1 < d_i$ for all $i=2, \ldots , s$. In the algorithm in \cite{hlr}, we assume that $C_1 \supset \cdots \supset C_s$, therefore the bound in (\ref{lowebound}) is sharp. Hence, if $C_1, \ldots , C_s$ are nested, both algorithms decode up to half of the minimum distance of the matrix-product code.  In the algorithm in \cite{hlr}, we run $DC_i$ $\binom{l}{s}$ times, for $i=1, \ldots , s$, in the worst-case. However, in the algorithm presented in this section, we only run $DC_i$ once, for $i=2, \ldots , s$ and we run $DC_1$ $l$ times, thus it will have polynomial complexity if the algorithms $DC_i$ have polynomial complexity, for $i = 1, \ldots , s$ (since $B$ is computed by solving a linear system). Hence the algorithm in \cite{hlr} becomes computationally intense for large values of $s,l$ but this algorithm does not.
\end{remark}

We can also consider this algorithm for correcting beyond the designed error-correction capability of $C$, if $l$ is even, $d_1$ is odd and $d_i > l d_1$, for $i=2,\dots,s$. Namely, the designed error correction capability of $C$ is $l t_1 + \lfloor \frac{l-1}{2} \rfloor = l t_1 + (l-2)/2$ and we consider now an error vector with  $wt (e) <l t_1 + l/2$, that is, we are correcting $1$ error beyond the designed error correcting capability of $C$. We should just modify line \ref{li:check} in Algorithm \ref{alg:dec1} to accept codewords $p$ with $wt(r'-p ) \le l t_1 + l/2$ and create a list with all the output codewords.

Again, we can decode $r v_i$ with the decoding algorithm $DC_i$ to obtain $c_i$, since$$wt \left(\sum_{j=1}^l v_{j,i} e_j \right)   \le wt (e) \le l t_1 +  \frac{l}{2}= \frac{l}{2} (2 t_1 + 1) \le \frac{l}{2} d_1 < \frac{d_i}{2}.$$ Moreover, there  exists $i \in \{ 1, \ldots , l\}$ such that $w(e_i) \le t_1$ as well because if $wt(e_i ) > t_1$ for all $i$ then $wt(e) \ge lt_1 + l >  l t_1 + l/2$. As before, we have $l$ candidates for $c_1$ and at least one of them is $c_1$, however now we cannot uniquely determine it: let $p=[c_1, \ldots , c_s]\cdot A$ and $p' = [c'_1, c_2 \ldots , c_s] \cdot A$ with $c_1 \neq c'_1$, one has that $d(p,r) \le l t_1 + l/2$ and $d(p',r) \ge l t_1 + l/2$. 

\begin{itemize}
\item $d( r,p)=  	 wt (e) \le l t_1 + l/2$.

\item $d( r,p')= wt ( a_{1,1}(c_1 -c'_1) + e_1, \ldots , a_{1,l}(c_1 -c'_1) + e_l) \ge l d_1 - wt (e) \ge l (2 t_1 +1) -(lt_1 l/2) = lt_1 + l/2$.
\end{itemize}The algorithm outputs $p$ and all the other codewords -obtained from the other $l-1$ candidates- that are at distance at most $l t_1 + l -1$ from $r$. As with $s=l=2$, the probability of having more than one codeword  in the output list is negligible, since $d(r,p') = lt_1+l/2$  if and only if the bound in  (\ref{lowebound}) is sharp, $d(c_1 , c'_1)=d_1$ and for every $j=1,\ldots, m$, $i=1,\ldots,l$, with $e_{j,i} \neq 0$, one has that $e_{j,i} = - a_{1,i}(c_{j,1} -c'_{j,1})$.

\begin{example}\label{ex:PrimeraExt}

Consider the following linear codes over $\mathbb{F}_3$,
\begin{itemize}
\item $C_1$ the $[26,20,4]$ cyclic code generated by $f_1=x^6 + x^5 + 2x^4 + 2x^3 + x^2 + x + 2$. 
\item $C_2$ the $[26,7,14]$ cyclic code generated by
$f_2=x^{19} + x^{18} + x^{17} + x^{15} + 2x^{14} + x^{13} + 2x^{12} + x^{11} + 2x^8 + 2x^7 + x^6
    + x^4 + x^3 + 2$. 
\item $C_3$ the
$[26,3,18]$ cyclic code generated by $f_3= x^{23} + 2x^{22} + x^{21} + 2x^{19} + 2x^{18} + x^{17} + x^{16} + x^{15} + x^{13} + x^{10} +2x^9 + x^8 + 2x^6 + 2x^5 + x^4 + x^3 + x^2 + 1$.
\end{itemize}
Let $C=[C_1C_2C_3]\cdot A$, where $A$ is the non-singular by columns matrix\begin{equation*} A=\left(
\begin{tabular}{cccc}
1 & 1 & 1 \\
0 & 1 & 2 \\
0 & 0 & 1 \\
\end{tabular}\right).
\end{equation*}

We use decoder $DC_i$ for $C_i$, which decodes up to half the minimum distance, i.e., $DC_1$, $DC_2$, $DC_3$ decode up to $t_1=1$, $t_2=6$ and $t_3=8$
errors, respectively. We have that $d_C =3 d_1= 12$ and since $A$ is triangular we have that the minimum distance of $C$ is $d(C) = d_C = 12$  and we may correct up to $t=5$ errors in a codeword of $C$. Note that $12=3d_1 \le d_2,d_3$. 

We consider now polynomial notation for codewords of $C_i$, for all $i$. Hence  the codewords of length $23$ in $C_i$ are polynomials in $\fq [x]/(x^{23}-1)$ and the words in $C$ are elements in $(\fq [x]/(x^{23}-1))^3$. Note that $C$ is a quasi-cyclic code. Let ${r}= {p} + {e}$ be the received word, with codeword ${p}=(0,0,0)$  and the error vector of weight $t=5$ $$ {e}=(e_1,e_2,e_3)=(1+x,2x^2+x^7,2x^{11}).$$

The matrix 
\begin{equation*} B=\left(
\begin{tabular}{cccc}
1 & 2 & 1 \\
0 & 1 & 1 \\
0 & 0 & 1 \\
\end{tabular}\right)
\end{equation*}
verifies that $AB=I_3$. Then $v_2$ and $v_3$ are the second and third columns of $B$ respectively.
Therefore $rv_2=c_2+2e_1+e_2=2+2x+2x^2+x^7$ and $rv_3=c_3+e_1+e_2+e_3=1+x+2x^2+x^7+2x^{11}$.

\begin{itemize}
\item[$\bullet$] We decode  $rv_3$ with $DC_3$ and we obtain $c_3=0$ because $wt(e_1+e_2+e_3)\leq wt(e)=5<t_2=6$. 

\item[$\bullet$] We decode $rv_2$ with  $DC_2$ and we obtain $c_2=0$ because $wt(2e_1+e_2)\leq wt(e)=5<t_2=6$.  

\item[$\bullet$] Subtracting $c_2$ and $c_3$ from $r=(c_1+e_1,c_1+c_2+e_2,c_1+2c_2+c_3+e_3)$ we get 
$r'=(c_1+e_1,c_1+e_2,c_1+e_3)$. Moreover we know that either $r_1'=c_1+e_1$ or $r_2'=c_1+e_2$
or $r_3'=c_1+e_3$ can be decoded with $DC_1$, so we should decode these three words. The weight of $r_3'$ 
is $1$, since the minimum distance of $C_1$ is $4$ there is only one codeword at distance $1$ of the zero-codeword,
and thus $c_3=0$. In the other two cases ($r_1'$ and $r_2'$) the weight is $2$, thus the output of the decoding algorithm $DC_1$ in both cases is either zero if it is the only codeword at distance $2$ (from $r_1'$ and $r_2'$ respectively) or ``failure" if there is more than one codeword at distance $2$.
\end{itemize}
\end{example}

\section{A decoding algorithm for Matrix-Product codes, second extension}\label{se:gene2}

In this section, we consider another extension of the algorithm in section \ref{se:2} for arbitrary $s\le l$. This algorithm imposes softer conditions (than the one in previous section) on the minimum distance of the constituent codes, however it can become computationally intense for large $s$ or $l$. Let $C = [C_1 \cdots C_s] \cdot A$ be a matrix-product code, we shall assume that $A$ is non-singular by columns and that $ d_i \ge i d_1$, for $i=2,\ldots,s$, where $d_i = d(C_i)$ is the minimum distance of $C_i$. 

The error-correction capability of $C_i$ is $t_i = \lfloor \frac{d_i -1 }{2} \rfloor$. From (\ref{lowebound}), one has that the designed minimum distance of $C$ is given by $d(C) \ge d_C =\min \{ld_1, (l-1) d_2, \ldots , (l - s +1)  d_s\}$ and it is computed in the following lemma.

\begin{lemma}
Let $C = [C_1 \cdots C_s] \cdot A$ be a matrix-product code, with $A$ non-singular by columns and $d_i \ge i d_1$, for $i=2,\ldots,s$. The designed minimum distance of $C$ is $d_C = l d_1$.
\end{lemma}

\begin{proof}
We claim that $l d_1 \le (l-i+1)d_i$, for $i =2, \ldots ,s$. Since $i d_1 \le d_i$, we have that $i (l-i+1)  d_1 < (l - i +1 )d_i$. Hence, $l d_1 \le  i (l-i+1) d_1 < (l-i+1) d_i$ if and only if $l \le i (l-i+1)$. One has that $$l \le i (l-i+1) \Longleftrightarrow l (i-1) \ge i^2 -i \Longleftrightarrow l \ge \frac{i^2 -i}{i-1} = i
$$
Thus, the claim holds since $i \le s \le l$.

Finally, we have that$$d_C=\min \{ld_1, (l-1) d_2, \ldots , (l - s +1)  d_s\} = l d_1.$$ 

\end{proof}

Hence, the designed error correction capability of the code $C$ is $$t=\left \lfloor\frac{l d_1 - 1}{2}\right \rfloor = \left\{ \begin{array}{ll} l t_1 + \lfloor \frac{l-1}{2} \rfloor & ~ \mathrm{if} ~ d_1 ~ \mathrm{is~odd}\\
l t_1 + l -1 & ~ \mathrm{if} ~ d_1 ~ \mathrm{is~even,}
\end{array} \right.$$ because $d_1 = 2 t_1 + 1$ if $d_1$ is odd and $d_1 = 2 t_1 + 2$ if $d_1$ is even. 

As in previous sections, we provide a decoding algorithm for the matrix-product code $C$, that decodes up to half of its designed minimum distance, assuming that we have a decoding algorithm $DC_i$ for $C_i$ which decodes up to $t_i$ errors, for $i=1,\ldots ,s$. A codeword in $C$ is an $m\times l$ matrix which has the form $ p = [c_1 \cdots c_s] \cdot A = (\sum_{j=1}^s a_{j,1} c_j, \ldots , \sum_{j=1}^s a_{j,l} c_j )$, where $c_j \in C_j$, for all $j$. Suppose that  $p$ is sent and that we receive $r =  p  +  e $, where $ e =(e_1,e_2,\ldots,e_l)$ is an error vector, an $m\times l$ matrix, with weight $wt (e) \le t$.

In order to decode $r$, we compute $c_i$, for $i=s, s-1, \ldots , 1$, inductively. Then, after $s$ iterations we compute $p$ by $p = [c_1 \cdots c_s] \cdot A$. We will now show how $c_i$ is obtained,  assuming that we have already obtained $c_s, c_{s-1}, \ldots, c_{i+1}$ (for $i=s$, we do not assume anything). Let $r^{i)}= (\sum_{j=1}^i a_{j,1} c_j + e_1, \ldots , \sum_{j=1}^i a_{j,l} c_j + e_l)$. We can obtain $r^{i)}$ from $r$ and $c_s, c_{s-1}, \ldots, c_{i+1}$, since $r^{s)}=r$ and  $r^{i)}=(r_1^{i+1)} - a_{i+1,1} c_{i+1}, \ldots , r_l^{i+1)} - a_{i+1,l} c_{i+1})$ for $i=s-1, \ldots, 1$.

Let $A_i$ be the submatrix of $A$ consisting of the first $i$ rows of $A$. Note that  $A_s=A$ and $A_i$ is an $i \times l$-matrix that  is non-singular by columns. Let $v^{i)} \in \fq^l$ such that $A v^{i)} = w_i = (0, \ldots, 0,1)^T \in \fq^i$. Such a $v^{i)}$ is not unique in general (it is only unique if $i=s=l$). For the sake of simplicity we will denote the coordinates of $v^{i)}$ by $v^{i)} = (v_1 , \ldots, v_l)$. Note that $v^{i)}$ is a solution of the corresponding linear system\begin{equation}\label{eq:sis}A_i x =w_i\end{equation}

Since $A v^{i)} = w_i$, we have that $[c_1 \cdots c_i] \cdot A_i v^{i)} = [c_1 \cdots c_i] w_i = c_i$. Hence, $r^{i)} v^{i)} = c_i + \sum_{j=1}^l v_j e_j$, in particular for $i=s$, we have $r v^{i)} = c_s + \sum_{j=1}^l v_j e_j$. We may decode $r^{i)}v^{i)}$ with $DC_i$ to obtain $c_i$ if $wt( \sum_{j=1}^l v_j e_j ) < d_i/2$. Therefore, it is wise to consider a vector $v^{i)}$ with low weight, that is with many coordinates equal to zero. 

We will consider a vector $v^{i)}$ with at least $l-i$ coordinates equal to zero, i.e. of weight $wt(v^{i)}) \le l - (l-i)=i$. Let $J=\{j_1 , \ldots , j_i \} \subset \{ 1, \ldots , l \}$ with $\# J=i$, we claim that we can compute $v^{i)}$, a solution of (\ref{eq:sis}), such that $v_j = 0$ for $j \not\in J$. Let $A_J$ be the $i\times i$-submatrix of $A_i$ given by $A_J = (a_{k,j})_{k\in \{1,\ldots,i\}, j\in J}$. Since $A$ is non-singular by columns, one has that $A_J$ is a full rank squared matrix. Let us consider the linear system\begin{equation}\label{eq:sis2}A_J x =w_i,\end{equation}where $x \in \mathbb{F}_q^i$. The linear system (\ref{eq:sis2}) has a unique solution.  Let $v^{i)}_J = (v_1 \ldots, v_l)$, where $v_{j_k} = x_k$, for  $k = 1, \ldots ,i$, and $v_j =0$ otherwise. Then, $v^{i)}_J$ is a solution of (\ref{eq:sis}) of weight lower than or equal to $i$, and the claim holds.

There are several choices for the set $J \subset \{ 1, \ldots , l \}$. We will prove in Theorem \ref{th:goodd} that at least for one choice of $J$, we will obtain $c_i$ by decoding $r^{i)}v^{i)}_J$ with $DC_i$. Therefore, in practice, we should consider $\binom{l}{i}$ vectors $\{ v^{i)}_J\}_{J \in \mathcal{J}}$, with $\mathcal{J} = \{ J\subset \{1,\ldots,l\}:  \#J=i \}$ and decode $r^{i)}v^{i)}_J$ with $DC_i$. We will have, at most, $\binom{l}{i}$ different candidates for $c_i$ and at least one of them will give $c_i$ as output.

In order to obtain $c_{i-1}$ we should iterate this process for every candidate obtained for $c_i$. After considering the previous computations for $i= s, s-1, \ldots , 1$, we may have several candidates for $[c_1, \ldots , c_s]$.  We can detect which candidate is equal to $p$ by checking that we have not corrected more than $\lfloor (d(C)-1)/2 \rfloor$ errors in total, that is, we check if $d(r - [c_1 \ldots , c_s]\cdot A) \le \lfloor (d(C)-1)/2 \rfloor$. The algorithm can be seen in procedural form in Algorithm \ref{alg:dec2}. However, it remains to prove that, at least for one choice of the set $J \subset \{1, \ldots , l \}$, one will obtain $c_i$. 

\begin{theorem}\label{th:goodd}

\end{theorem} 
Let $e$ with $wt (e)  \le t$. There exists $J \subset \{1, \ldots , l\}$, with $\# J =i$, such that $\sum_{j \in J} wt(e_{j}) < d_i/2$, for $i=1, \ldots , s$.
\begin{proof}  
Let  $v^{i)}_J = (v_1, \ldots ,v_l)$ as before.  We have that, 
$$wt\left( \sum_{j=1}^l v_j e_j \right) \le wt \left( \sum_{j \in J} e_{j} \right) \le  \sum_{j \in J} wt(e_{j}). $$

The result claims that there exists $J \subset \{1, \ldots , l\}$, with $\# J =i \in \{2, \ldots , s\}$, such that $\sum_{j \in J} wt(e_{j}) < d_i/2$. Let $\mathcal{J} = \{ J\subset \{1,\ldots,l\}:  \#J=i \}$, and let us assume that the claim does not hold. We consider every $\binom{l}{i}$ possible subset $J \subset \{1, \ldots, l\}$ with $i$ elements, then$$ \sum_{J \in \mathcal{J}} \sum_{j\in J} wt(e_j)  \ge \binom{l}{i} \frac{d_i}{2}.$$Moreover, since $\binom{l-1}{i-1}$ sets of $\mathcal{J}$ contain $j$, for $j \in \{1,\ldots, l\}$, we have $$  \sum_{J \in \mathcal{J}} \sum_{j\in J} wt(e_j) = \sum_{j=1}^l  \binom{l-1}{i-1} wt(e_j) =   \binom{l-1}{i-1} wt(e) < \binom{l-1}{i-1} \frac{l d_1}{2}.$$Which implies that $$ \binom{l}{i} d_i < \binom{l-1}{i-1} l d_1,$$ therefore $i d_1 > d_i$, contradiction.

For $i=1$, we have that  $\mathcal{J} = \{ \{1\}, \ldots , \{l\} \}$.Hence, we have that $r^{1)} v_J = c_j + e_j$, for $J=\{j\}$. The result claims that there exists $j \in \{1 , \ldots , l\}$, such that $wt (e_j) < d_1 / 2$. Otherwise, $wt (e) \ge l t_1 + l > t$, which is a contradiction.
 \end{proof}

\newcommand{\decc}{\ensuremath{\mbox{\sc A decoding algorithm for $C=[C_1 \cdots C_s] \cdot A$, second extension}}}
\begin{algorithm}[h!]
  \caption{\decc}
\algsetup{indent=2em}
  \begin{algorithmic}[1]\label{alg:dec2}
    \REQUIRE Received word $r = p + e$ with $c \in C$ and $wt(e) < d(C)/2$. Where $d_i = d(C_i)$ with $i d_1 < d_i$ and $A$  a non-singular by columns matrix. Decoder $DC_i$ for code $C_i$, $i = 1,\ldots ,s$.
    \ENSURE $p$.
    \medskip
    \STATE $r'=r$;
    \STATE $Candidates'=\{[0 \cdots 0 ]\}$ ($ 0 \in \mathbb{F}_q^m$); 
    \FOR{$i=s, s-1, \ldots, 2,1$}
        \STATE $Candidates = Candidates'$;
        \STATE $Candidates' = \{\}$;
        \FOR{$c=(c_1,\ldots,c_s)$ in $Candidates$}
            \STATE $r=(r'_1 - \sum_{j=i+1}^s a_{j,1} c_j, \ldots , r'_l - \sum_{j=i+1}^s a_{j,l} c_j)$;
            \FOR{$J \subset \{1,\ldots , l\}$ with $\# J =i$}\label{li:cc}
                 \STATE Solve linear system $A_J x = w_i$;
                 \STATE $v=(0,\ldots,0)$;
                 \FOR{$k=1\ldots , i$} 
                     \STATE $v_{j_k} = x_k$;
                 \ENDFOR
                 \STATE $b_i = DC_i (rv)$;
                 \IF{$b_{i}=$ {\it ``failure"}}
                       \STATE {\it Break} the loop and consider another $J$ in line \ref{li:cc};
                 \ENDIF          
                \STATE $Candidates' = Candidates' \cup \{ [0   \cdots 0 b_i c_{i+1} \cdots c_s] \}$;                   
             \ENDFOR
        \ENDFOR
    \ENDFOR
  
    \FOR{$c$ in $Candidates'$} 
       \STATE $p=[c_1\cdots c_s] \cdot A$;
        \IF{$p \in C$ and $wt(r-p ) \le \lfloor (d(C)-1)/2 \rfloor $ }\label{li:check2}
            \STATE RETURN: $p$;
    \ENDIF
  \ENDFOR
  \end{algorithmic}
\end{algorithm}

\begin{remark}
Let us compare this decoding algorithm to the algorithm in \cite{hlr}. In the algorithm in \cite{hlr}, we assume that $C_1 \supset \cdots \supset C_s$, $A$ is non-singular by columns and in the worst-case we run $DC_i$ $\binom{l}{s}$ times, for $i=1, \ldots , s$. For the algorithm in this section, we assume that $d_i \ge i d_1$ for all $i=2, \ldots , s$, $A$ is also non-singular but in the worst-case we run $DC_i$  $\prod_{j=i}^s \binom{l}{j}$ times. Thus, the algorithm presented in this section can become computationally intense for large values of $s,l$. If $C_1 , \ldots , C_s$ are nested, both algorithms decode up to half of the minimum distance of the code, since the bound in (\ref{lowebound}) is sharp. 
\end{remark}

As in previous sections, one can also consider this algorithm for correcting beyond the designed error-correction capability of $C$, if $l$ is even, $d_1$ is odd and $d_i > i d_1$. Namely, the designed error correction capability of $C$ is $l t_1 + \lfloor \frac{l-1}{2} \rfloor = l t_1 + (l-2)/2$ and we consider now an error vector with  $wt (e) <l t_1 + l/2$, that is, we are correcting $1$ error beyond the error correcting capability of $C$. We should just modify line \ref{li:check2} of Algorithm \ref{alg:dec2} to accept codewords $p$ with $wt(r'-p ) \le l t_1 + l/2$ and create a list with all the output codewords.

We shall prove that, at least for one choice of the set $J \subset \{1, \ldots , l \}$, one will again obtain $c_i$. 

\begin{theorem}\label{th:goodd2}

\end{theorem} 
Let $e$ with $wt (e)  \le l t_1 + l/2 $, with $d_1$ odd, $l$ even and $d_i > i d_1$. There exists $J \subset \{1, \ldots , l\}$, with $\# J =i$, such that $\sum_{j \in J} wt(e_{j}) < d_i/2$, for $i=1, \ldots , s$.
\begin{proof}  
Let  $v^{i)}_J = (v_1, \ldots ,v_l)$ as before.  We have that, 
$$wt\left( \sum_{j=1}^l v_j e_j \right) \le wt \left( \sum_{j \in J} e_{j} \right) \le  \sum_{j \in J} wt(e_{j}). $$

The result claims that there exists $J \subset \{1, \ldots , l\}$, with $\# J =i \in \{2, \ldots , s\}$, such that $\sum_{j \in J} wt(e_{j}) < d_i/2$. Let $\mathcal{J} = \{ J\subset \{1,\ldots,l\}:  \#J=i \}$, and let us assume that the claim does not hold. We consider every $\binom{l}{i}$ possible subsets $J \subset \{1, \ldots, l\}$ with $i$ elements, then$$ \sum_{J \in \mathcal{J}} \sum_{j\in J} wt(e_j)  \ge \binom{l}{i} \frac{d_i}{2}.$$Moreover, since $\binom{l-1}{i-1}$ sets of $\mathcal{J}$ contain $j$, for $j \in \{1,\ldots, l\}$, we have $$  \sum_{J \in \mathcal{J}} \sum_{j\in J} wt(e_j) = \sum_{j=1}^l  \binom{l-1}{i-1} wt(e_j) =   \binom{l-1}{i-1} wt(e) \le \binom{l-1}{i-1} (lt_1 + \frac{l}{2}).$$Which implies that, $$ \binom{l}{i} \frac{d_i}{2} \le \binom{l-1}{i-1} \frac{l}{2}(2t_1 +1),$$ therefore $d_i \le  i (2t_1 +1) = i d_1$, contradiction.

For $i=1$, we have that  $\mathcal{J} = \{ \{1\}, \ldots , \{l\} \}$. Therefore, we have that $r^{1)} v_J = c_j + e_j$, for $J=\{j\}$. The result claims that there exist $j \in \{1 , \ldots , l\}$, such that $wt (e_j) < d_1 / 2$. Otherwise, $wt (e) \ge l t_1 + l >  l t_1 + l/2 $, which is a contradiction.
 \end{proof}

This algorithm will output  a list containing the sent word, however it cannot be uniquely determined: let $p=[c_1, \ldots , c_s]\cdot A$ and $p' = [c'_1, c_2 \ldots , c_s]\cdot A$ with $c_1 \neq c'_1$, we claim that $d(p,r) \le l t_1 + l/2$ and $d(p',r) \ge l t_1 + l/2$.

\begin{itemize}
\item $d( r,p)=  	 wt (e) \le l t_1 + l/2$.

\item $d( r,p')= wt ( a_{1,1}(c_1 -c'_1) + e_1, \ldots , a_{1,l}(c_1 -c'_1) + e_l) \ge l d_1 - wt (e) \ge l (2 t_1 +1) -(lt_1 l/2) = lt_1 + l/2$.
\end{itemize}The algorithm outputs $p$ and all the other codewords -obtained from the other candidates- that are at distance at most $l t_1 + l -1$ from $r$. As with $s=l=2$, the probability of having more than one codeword in the output list is negligible, since $d(r,p') = lt_1+l/2$  if and only if the bound in  (\ref{lowebound}) is sharp, $d(c_1 , c'_1)=d_1$ and for  every $j=1,\ldots, m$, $i=1,\ldots,l$, with $e_{j,i} \neq 0$, one has that $e_{j,i} = - a_{1,i}(c_{j,1} -c'_{j,1})$.

\begin{example}
Consider the following linear codes over $\mathbb{F}_3$,
\begin{itemize}
\item $C_1$ the $[26,16,6]$ cyclic code generated by $f_1=x^{10} + 2x^7 + 2x^4 + x^3 + 2x^2 + x + 2$. 
\item $C_2$ the $[26,7,14]$ cyclic code generated by
$f_2=x^{19} + x^{18} + x^{17} + x^{15} + 2x^{14} + x^{13} + 2x^{12} + x^{11} + 2x^8 + 2x^7 + x^6
    + x^4 + x^3 + 2$. 
\item $C_3$ the
$[26,3,18]$ cyclic code generated by $f_3= x^{23} + 2x^{22} + x^{21} + 2x^{19} + 2x^{18} + x^{17} + x^{16} + x^{15} + x^{13} + x^{10} +2x^9 + x^8 + 2x^6 + 2x^5 + x^4 + x^3 + x^2 + 1$.
\end{itemize}
Let $C=[C_1C_2C_3]\cdot A$, where $A$ is the non-singular by columns matrix \begin{equation*} A=\left(
\begin{tabular}{ccc}
1 & 1 & 1 \\
0 & 1 & 2 \\
1 & 0 & 1 \\
\end{tabular}\right),
\end{equation*}we consider again polynomial notation for $C$ (see example \ref{ex:PrimeraExt}). We use decoder $DC_i$ for $C_i$, which decodes up to half the minimum distance, i.e., $DC_1$, $DC_2$, $DC_3$ decode up to $t_1=2$, $t_2=6$ and $t_3=8$
errors, respectively. Note that $2 d_1 = 12 \le 14 = d_2$ and $3d_1 = 18 \le 18 = d_3$. We have that 
$d(C)= d_C= 3 d_1= 18$. Therefore we may correct up to $t=8$ errors in a codeword of $C$. 

Let ${r}= {p} + {e}$ be the received word, with codeword ${p}=(0,0,0)$  and the error vector of weight $t=8$ $$ {e}=(e_1,e_2,e_3)=(1+x+x^2,1+2x^2+x^7,x^5+2x^{11}).$$

We solve the system 
\begin{equation*} \left(
\begin{matrix}
1 & 1 & 1 \\
0 & 1 & 2 \\
1 & 0 & 1 \\
\end{matrix}\right)\left(\begin{matrix} 
x  \\
y  \\
z  \\
\end{matrix}\right)=\left(\begin{matrix}
0  \\
0  \\
1  \\
\end{matrix}\right).
\end{equation*}which has solution $(2,2,2)^T$. Set $r^{3)}=r$ and $v^{3)}_{\{1,2,3\}}=(2,2,2)$. Therefore $r^{3)}v^{3)}_{\{1,2,3\}}= c_3 + 2e_1 +2 e_2 + 2e_3$. Since $DC_3$ can correct up to $8$ errors and $wt(-e_1-e_2-e_3) =8$, we have $$DC_3(r^{3)}v^{3)}_{\{1,2,3\}})=c_3=0$$

Removing $c_3$ in $r^{3)}$, we obtain $r^{2)} =(r_1^{3)} - c_{3}, \ldots , r_l^{3)} - c_{3}) = r$. Since there are $3$ possible sets, $\{1,2\}, \{1,3\}, \{2,3\} \subset \{ 1,2,3 \}$ , with $2$ elements, we solve the corresponding systems of equations give by (\ref{eq:sis2}):
 
\begin{equation*} \left(
\begin{matrix}
1 & 1  \\
0 & 1  \\
\end{matrix}\right)\left(\begin{matrix}
x  \\
y  \\
\end{matrix}\right)=\left(\begin{matrix}
0  \\
1  \\
\end{matrix}\right),
\end{equation*}

\begin{equation*} \left(
\begin{matrix}
1 & 1  \\
0 & 2  \\
\end{matrix}\right)\left(\begin{matrix}
x  \\
y  \\
\end{matrix}\right)=\left(\begin{matrix}
0  \\
1  \\
\end{matrix}\right),
\end{equation*}

\begin{equation*} \left(
\begin{matrix}
1 & 1  \\
1 & 2  \\
\end{matrix}\right)\left(\begin{matrix}
x  \\
y  \\
\end{matrix}\right)=\left(\begin{matrix}
0  \\
1  \\
\end{matrix}\right).
\end{equation*}

These systems have solution $(0,1)^T$, $(0,2)^T$ and $(2,1)^T$ respectively. Therefore, $v^{2)}_{1,2}=(0,1,0)$, $v^{2)}_{1,3}=(0,0,2)$ and $v^{2)}_{2,3}=(0,2,1)$. Thus $r^{2)}v^{2)}_{\{1,2\}}=c_2+e_2$, $r^{2)}v^{2)}_{\{1,3\}}=c_2+2e_3$ and $r^{2)}v^{2)}_{\{2,3\}}=c_2+2e_2+e_3$. Since $t_2=6$ and $wt(e_2)=3\le 6$, $wt(2e_3)=2\le 6$ and $wt(2e_2+e_3)\leq wt(e_2)+wt(e_3)=5 \le 6$, we have  $$DC_2(r^{3)}v^{3)}_J)=c_2=0, \mathrm{for~} J=\{1,2\}, \{1,3\}, \{2,3\}.$$

Therefore, we only have one candidate for $c_2$. Removing $c_2$ in $r^{2)}$, we obtain $r^{1)} =(r_1^{2)} - c_{2}, \ldots , r_l^{2)} - c_{2}) = r$. Since there are $3$ possible sets, $\{1\}, \{2\}, \{3\} \subset \{ 1,2,3 \}$, with $1$ element, we solve the corresponding systems of equations give by (\ref{eq:sis2}). In this case the 3 systems of equations are
  
\begin{equation*} \left(
\begin{matrix}
1  \\
\end{matrix}\right)\left(\begin{matrix}
x  \\
\end{matrix}\right)=\left(\begin{matrix}
1  \\
\end{matrix}\right).
\end{equation*}

Thus, the solution is $(1)$ and  $v^{1)}_{\{1\}}=(1,0,0)$, $v^{1)}_{\{2\}}=(0,1,0)$ and $v^{1)}_{\{3\}}=(0,0,1)$. Thus 
 $r^{1)}v^{1)}_{\{1\}}=c_1+e_1$, $r^{1)}v^{1)}_{\{2\}}=c_1+e_2$ and $r^{1)}v^{1)}_{\{3\}}=c_1+e_3$. We consider $DC_1(r^{3)}v^{3)}_J)$: we obtain ``failure" for $DC_1(r^{3)}v^{3)}_{\{1\}})$ and $DC_1(r^{3)}v^{3)}_{\{2\}})$ since $e_1$ and $e_2$ have weight $3$ and there is no codeword at distance $2$ because $C_1$ has minimum distance $6$. One has that $wt(e_3) = 2 \le t_1$, therefore $$DC_1(r^{3)}v^{3)}_{\{3\}}) = c_1 = 0$$

Finally we get $p = [c_1 c_2 c_3] \cdot A = (0,0,0)$.
\end{example}

\section*{Acknowledgements}
This paper was written in part during a visit of the second author to the Mathematics department of San Diego State University. He wishes to thank this institution and the first author for hospitality.
 
\bibliographystyle{plain}
\bibliography{HRarxiv}

\end{document}